\documentclass[conference, 10 pt]{ieeeconf}  

\IEEEoverridecommandlockouts                              
\usepackage{float}
\usepackage{amsmath}

\usepackage{amssymb}

\usepackage{amsthm}
\usepackage{algorithm} 
\usepackage{algpseudocode} 
\usepackage{xcolor}
\usepackage{bbm}     
\usepackage{graphicx} 
\usepackage[normalem]{ulem}
\newtheorem{theorem}{Theorem}
\newtheorem{lemma}{Lemma}

\newtheorem{assumption}{Assumption} 
\newtheorem{remark}{Remark} 
 
\usepackage{dsfont}
\usepackage{booktabs} 
\usepackage[utf8]{inputenc}

\renewcommand{\baselinestretch}{.96}

\title{\LARGE \bf A Koopman Set-Membership Approach for Nonlinear Data-Driven Control with Stability Guarantees}

\author{Yifan Xie, Zuxun Xiong,  Julian Berberich, Antonis Papachristodoulou, Frank Allg\"{o}wer 
\thanks{F. Allg\"{o}wer is thankful that his work was funded by Deutsche Forschungsgemeinschaft (DFG, German Research
Foundation) under Germany’s Excellence Strategy - EXC 2075 - 390740016 and within grant AL 316/15-1-468094890.
A. Papachristodoulou was supported in part by UK's Engineering, Physical Sciences Research Council projects EP/Y014073/1 and in part by the EPSRC under project UKRI2108.
The authors thank the International Max Planck Research School
for Intelligent Systems (IMPRS-IS) for supporting Yifan Xie.
	}
\thanks{Y.~Xie, J.~Berberich, F. Allg\"{o}wer are with the Institute for Systems Theory and Automatic Control, University of Stuttgart, 70550 Stuttgart, Germany. Julian~berberich is also with the Center for Integrated Quantum Science and Technology (IQST), University of Stuttgart, 70550 Stuttgart, Germany.
{\tt\small  \{yifan.xie, julian.berberich,
        frank.allgower\}@ist.uni-stuttgart.de}}
\thanks{Z. Xiong and A. Papachristodoulou are with Department of Engineering Science, University of Oxford, Parks Road, Oxford, OX1 3PJ, U.K. {\tt\small \{zuxun.xiong, antonis\}@eng.ox.ac.uk }}
}
\begin{document}

\maketitle
\thispagestyle{empty}
\pagestyle{empty}

\begin{abstract}
This paper proposes a data-driven controller design method for unknown nonlinear systems based on a Koopman bilinear realization. 
Using Koopman operator theory, the nonlinear system can be represented as a bilinear discrete-time system with a residual error term. 
The residual error is proportionally bounded by the norm of the lifted state and input, while the system matrices of the bilinear model are unknown.
Assuming that bounds on the residual error are available, the unknown system matrices are characterized via a set-membership representation using the collected input–state data pairs of the nonlinear system. 
A data-driven controller design method is proposed to ensure stability for all bilinear systems within this set-membership description and for all admissible residual errors.
More specifically, we design a rational state-feedback controller that stabilizes the bilinear model with residual error and, consequently, the original nonlinear system, by solving a sum-of-squares (SOS) program. 
The effectiveness of the proposed approach is demonstrated through numerical examples.
\end{abstract}

\section{Introduction}\label{sec:1}
The design of stabilizing controllers for nonlinear systems is a challenging problem in control theory~\cite{isidori1985nonlinear}. 
Traditional controller design methods often rely on local linearization techniques, which require exact knowledge of the system dynamics, and stability cannot be guaranteed over a broad region. 
The development of convex optimization methods has re-opened the controller synthesis problem for nonlinear systems. 
Specifically, the sum-of-squares (SOSs) method transforms polynomial positivity conditions into convex optimization problems through tractable relaxations~\cite{parrilo2000structured}. 
Existing SOS-based control design methods generally fall into two categories: those relying on specific structural properties of the system to achieve convexification~\cite{prajna2004nonlinear}, and those utilizing iterative algorithms to optimize the controller and the Control Lyapunov Function iteratively~\cite{newton2025design}. 
An alternative paradigm for nonlinear control design uses Koopman operator theory~\cite{koopman1931hamiltonian}. 
By lifting the finite-dimensional nonlinear state space into an infinite-dimensional space, the Koopman operator approximates the nonlinear dynamics globally as a linear system. 
This transformation enables the direct application of well-established linear control methods. 
Recent studies indicate that Koopman linearization cannot adequately capture the complex coupling between states and inputs, even with an infinite-dimensional lifted state. 
Instead, Koopman bilinearization provides a significantly more accurate representation for nonlinear control-affine systems compared with Koopman linearization~\cite{strasser2026overview}. Crucially, the residual error introduced by a Koopman bilinear realization can be rigorously bounded~\cite{strasser2026safedmd,philipp2025error}, which paves the way for Koopman-based nonlinear controller designs that provide strict theoretical guarantees for stability~\cite{bold2024data, strasser2025koopman}.

These methods, however, all require first identifying a linear or bilinear model using data and then designing controllers based on the identified system dynamics. 
To address this, the direct data-driven control approach aims to synthesize controllers directly from collected data without an intermediate system identification. 
A branch of this field is based on Willems' fundamental lemma~\cite{willems2005note}, which characterizes system behaviour using persistently exciting data trajectories~\cite{de2019formulas}. 
While extending this behavioral framework to nonlinear dynamics is difficult, recent progress has been made using online local linear approximations for data-driven model predictive control (MPC)~\cite{berberich2022linear}. 
Other recent works explore this extension via the Koopman operator~\cite{shang2024willems,xiong2025data}.
Nevertheless, due to strong underlying assumptions, using these methods to rigorously handle Koopman approximation errors remains a major challenge. 
An alternative direct data-driven control framework called data informativity has been proposed using robust control theory~\cite{van2023quadratic,DBLSCT2025}. 
Under this framework, quadratic matrix inequalities are typically used to characterize a set of systems that explains the measured data based on assumptions on the noise such as energy or instantaneous noise bounds. 
Data-driven controller designs aim to find a controller that stabilizes all systems that are consistent with the data.
This framework was first proposed for linear time-invariant systems \cite{van2023quadratic}, and now has been extended to various kinds of nonlinear systems, such as linear parameter-varying systems~\cite{verhoek2024decoupling}, polynomial systems~\cite{martin2023inference,guo2021data}, and bilinear systems~\cite{xie2025bilinear}.
Various control objectives are considered including $H_2$ control, $H_\infty$ control \cite{DBLSCT2025}, and min-max MPC~\cite{xie2026data, nguyen2023lmirobust}.
However, data-driven controller design for more general nonlinear systems remains unexplored in this framework.

In this work, we propose a controller design method for unknown continuous-time nonlinear control-affine systems using only input-state data.
The nonlinear system can be represented by a bilinear model with a residual term based on Koopman operator theory.
The bound on the residual is proportional to the norm of the lifted states and inputs.
With this knowledge of residual bounds, we explicitly construct a consistent set of all admissible system matrices using the input-state data. 
We formulate an SOS program that directly synthesizes a rational state-feedback control law for all possible bilinear systems within the consistent set.
We show that the continuous-time nonlinear system can be exponentially stabilized by a sampled controller designed for the bilinear systems.
Numerical examples show that the  proposed controller design method is effective and illustrate the influence of some parameters on computational complexity and convergence rate.

The rest of this paper is organized as follows. 
In Section \ref{sec:2}, we introduce the considered nonlinear system and the problem setup. 
In Section \ref{sec:3}, we first present the data-driven characterization of the bilinear system.
Then, we propose the data-driven controller design method and prove exponential stability of the unknown nonlinear system.
We illustrate the effectiveness of the proposed method using numerical examples in Section \ref{sec:4} and draw conclusions in the last section. 

\textbf{Notation:} 
We use $\mathrm{diag}$ and $\otimes$ to denote block diagonal matrices and Kronecker product, respectively.  
For matrices $A$ and $B$ of compatible dimensions, we abbreviate $ABA^\top$ to $AB\begin{bmatrix}\star\end{bmatrix}^\top$.
The set of integers from $a$ to $b$ is denoted by $\mathbb{I}_{[a, b]}$.
For a symmetric matrix $P$, we write $P\succ 0$ or $P\succeq 0$ if $P$ is positive definite or positive semi-definite, respectively, where negative (semi-)definite matrices are defined analogously.
We denote by $\mathbb{R}[x]_{\alpha}$ the set of all polynomials in the variable $x\in\mathbb{R}^n$ with degree $\alpha$.
The set of all $m\times n$ matrices with elements in $\mathbb{R}[x]_{\alpha}$ is denoted by $\mathbb{R}[x]^{m\times n}_{\alpha}$.
A matrix $S\in\mathbb{R}[x]^{n\times n}_{2\alpha}$ is an SOS matrix in $x$ if it can be decomposed as $S=T^\top T$ for some $T\in\mathbb{R}[x]^{m\times n}_{\alpha}$ and we write it as $S\in \text{SOS}[x]^n_{2\alpha}$.
A matrix $S\in\mathbb{R}[x]^{n\times n}_{2\alpha}$ is strictly SOS if $(S-\epsilon I_n)\in \text{SOS}[x]^n_{2\alpha}$ for some $\epsilon>0$ and we write it as $S\in \text{SOS}_+[x]^n_{2\alpha}$.


\section{Problem Setup}\label{sec:2}

We consider an unknown nonlinear control-affine system
\begin{equation}\label{system_nonlinear}
    \dot{x}(t)=f(x(t))+\sum_{i=1}^m g_i(x(t))u_i(t),
\end{equation}
where $x(t)\in\mathbb{R}^n$ is the state and $u_i(t)$ is the $i$-th element of input vector $u(t)\in\mathbb{R}^m$.
We assume that $f$ and $g_i, \forall i\in\mathbb{I}_{[1, m]}$ are unknown, and $f(0)=0$, i.e., $(x, u)=(0, 0)$ is an equilibrium of the system.
Based on the Koopman operator theory and its finite-dimensional approximation, the following bilinear discrete-time dynamics can represent the nonlinear system \eqref{system_nonlinear} 
\begin{equation}\label{system_bilinear}
\Phi(x_{t+1})\!=\!A_s\Phi(x_t)\!+\!B_su_t\!+\!\tilde{B}_s(u_t\otimes \Phi(x_t))\!+\!r(x_t, u_t),
\end{equation}
where $x_t=x(t\Delta k)$ for some sampling time $\Delta k>0$ and $t\in\mathbb{N}$, $\Phi(x)=\begin{bmatrix}x^\top &\phi_{n+1}(x) &\ldots &\phi_N(x)\end{bmatrix}$ is the lifted state with $\phi_i:\mathbb{R}^n\rightarrow\mathbb{R}$ for all $i\in\mathbb{I}_{[n+1, N]}$, and 
$r(x, u)$ is the residual.
In this paper, we assume that the matrices $A_s, B_s, \tilde{B}_s$ are unknown and the following assumption on the residual $r(x, u)$ and the lifting function is satisfied.

\begin{assumption}\upshape\label{assumption1}
Given compact sets $\mathbb{X}\in\mathbb{R}^n$ and $\mathbb{U}\in\mathbb{R}^m$, the residual function $r:\mathbb{X}\times \mathbb{U}\rightarrow \mathbb{R}^N$ satisfies the following proportional bound
\begin{equation}\label{residual}
    \|r(x, u)\|\leq c_x\|\Phi(x)\|+c_u\|u\|
\end{equation}
for some known $c_x, c_u>0$ and for all $x\in\mathbb{X}$, $u\in\mathbb{U}$.
We denote the set of functions $r$ satisfying \eqref{residual} by $Re$.
Besides, we assume that the lifting function $\Phi:\mathbb{R}^n\rightarrow\mathbb{R}^N$ is known and satisfies $\Phi\in \mathcal{C}^1(\mathbb{R}^n, \mathbb{R}^N)$ with $\Phi(0)=0$.
\end{assumption} 

\begin{remark}\upshape
According to existing studies,  there is an approximation error of the bilinear model \eqref{system_bilinear} arising from the use of a finite-dimensional dictionary for the lifted state, and the data-based estimation.
This approximation error can be bounded proportionally to the norms of the lifted state and the input. Such a bound is crucial for the robust controller design of the original nonlinear system, as the control objective must ensure robustness against all admissible errors within this bound.
These types of proportional error bounds have been explicitly characterized in data-driven Koopman-based frameworks, such as SafeDMD~\cite{strasser2026safedmd}, and EDMD~\cite{bold2024data}.
For a comprehensive overview of Koopman-based control and the associated approximation errors, the interested reader is referred to \cite{strasser2026overview}.
In this paper, we assume that the residual satisfies a proportional bound of the form in Assumption~\ref{assumption1}, and that the corresponding coefficients can be obtained or estimated over the considered compact sets.
\end{remark}

We assume that a number of input and state pairs (data) generated from the system \eqref{system_nonlinear} are available:
\begin{equation}\label{data1}\nonumber
    \left \{(x_{t_i}^f, u_{t_i}^f, x_{t_i+\Delta k}^f)_{i=1}^{T_f}: x_{t_i}^f, x_{t_i+\Delta k}^f\in\mathbb{X}, u_{t_i}^f\in\mathbb{U} \right\}.
\end{equation}
Using the known lifting function, we obtain the following lifted state and input pairs for the bilinear system~\eqref{system_bilinear}
\begin{equation}\label{data2}
    \left \{\!(\Phi(x_{t_i}^f), u_{t_i}^f, \Phi(x_{t_i+\Delta k}^f)_{i=1}^{T_f}): x_{t_i}^f, x_{t_i+\Delta k}^f\in\mathbb{X}, u_{t_i}^f\in\mathbb{U}\!\right \}.
\end{equation}
The corresponding residual is unknown, but satisfies the bound in Assumption~\ref{assumption1}.
In this paper, we aim to directly design a controller using the data \eqref{data2} such that the resulting closed-loop system~\eqref{system_nonlinear} is stabilized to the origin.




\section{Koopman-based Data-driven Controller Design}\label{sec:3}

In Section~\ref{sec3.1}, we first propose a data-driven characterization of the set of consistent system matrices for the bilinear system~\eqref{system_bilinear}, using the data \eqref{data2} and Assumption~\ref{assumption1}.
Then, an SOS program is proposed to derive a rational state-feedback controller for all bilinear systems within the consistent set in Section~\ref{sec3.2}.
We prove that the resulting nonlinear system is exponentially stabilized based on the sampled rational state-feedback controller designed for the bilinear systems.

\subsection{Data-driven characterization}\label{sec3.1}
The set of $(A, B, \tilde{B})$ that is consistent with the data~\eqref{data2} and  Assumption~\ref{assumption1} is defined as
\[\Sigma=\left \{(A, B, \tilde{B}): (A, B, \tilde{B})\in\Sigma_i, \forall i\in\mathbb{I}_{[1, T_f]}\right \},\]
where $\Sigma_i$ is the set of system matrices $(A, B, \tilde{B})$ for which there exist a residual function satisfying Assumption \ref{assumption1} explaining the data $\Phi_{t_i}^f, u_{t_i}^f, \Phi_{t_i+\Delta k}^f$, defined by
\begin{equation}\nonumber
\Sigma_{i}\!=\!\!\left\{\!(A, \!B, \!\tilde{B})\!:\!
\begin{gathered}
r(x_{t_i}^f, u_{t_i}^f)\!=\!\Phi(x_{t_i+\Delta k}^f)\!-\!A\Phi(x_{t_i}^f)\!-\!Bu_{t_i}^f\!\\-\!\tilde{B}(u_{t_i}^f\otimes \Phi(x_{t_i}^f)) 
\text{ holds for some }r\in Re
\end{gathered}
\right\}.
\end{equation}

In the following lemma, we characterize the set of $(A, B, \tilde{B})$ consistent with the data \eqref{data2} by a quadratic matrix inequality.

\begin{lemma}\upshape
Suppose Assumption~1 holds. The set of system matrices consistent with the data~\eqref{data2} is equal to
\begin{equation}\label{Sigma}
\left\{\!\!(A, B, \tilde{B}):\!\!\!
\begin{gathered}
\begin{bmatrix}
I \!\!\!&A \!\!\!&B \!\!\!&\tilde{B}
\end{bmatrix}
\!\!M(\tau(z))\!\!
\begin{bmatrix}
I \!\!\!&A \!\!\!&B \!\!\!&\tilde{B}
\end{bmatrix}^\top\!\succeq\! 0,  \\
\forall \tau(z)\!=\!(\tau_1(z), \ldots, \tau_{T_f}(z)), i\in\mathbb{I}_{[1, T_f]}\\ \tau_i\in \text{SOS}[z]_{2\alpha}, z\in\mathbb{R}^N
\end{gathered}
\right\},
\end{equation}
where
$M(\tau(z))\!=\!\!\!\sum_{i=1}^{T_f}\tau_i(z)N_i$ and 
\[N_i\!=\!\!\!\begin{bmatrix}
I \!\!\!&\Phi(x_{t_{i}+\Delta k}^f)\\
0 \!\!\!&-\Phi(x_{t_i}^f)\\
0 \!\!\!&-u_{t_i}^f\\
0 \!\!\!&-(u_{t_i}^f \!\otimes\! \Phi(x_{t_i}^f))
\end{bmatrix}\!\!\!
\begin{bmatrix}
    (c_x\|\Phi(x_{t_i}^f)\|\!\!+\!c_u\|u_{t_i}^f\|)^2 \!\!\!\!\!&0\\
    0 \!\!\!\!\!&-I
\end{bmatrix}\!\!\!
\begin{bmatrix}
\star
\end{bmatrix}^\top.\]
\end{lemma}
\begin{proof}
Since the residual function $r$ satisfies Assumption~\ref{assumption1}, we have $\|r(x_{t_i}^f, u_{t_i}^f)\|\leq c_x\|\Phi(x_{t_i}^f)\|+c_u\|u_{t_i}^f\|$.
Replacing $r(x_{t_i}^f, u_{t_i}^f)$ by $\Phi(x_{{t_i}+\Delta k}^f)-A\Phi(x_{t_i}^f)-Bu_{t_i}^f-\tilde{B}(u_{t_i}^f\otimes \Phi(x_{t_i}^f))$ in the previous bound, we can show that $\Sigma_i$ is equal to
\begin{equation}\label{sigma_i}
\left\{\begin{gathered}
(A, B, \tilde{B}):
\begin{bmatrix}
I \!\!\!&A \!\!\!&B \!\!\!&\tilde{B}
\end{bmatrix}
N_i
\begin{bmatrix}
I \!\!\!&A \!\!\!&B \!\!\!&\tilde{B}
\end{bmatrix}^\top\!\succeq\! 0
\end{gathered}
\right\},
\end{equation}
Similar to \cite{bisoffi2021trade} and using the characterization of the set $\Sigma_i$ in \eqref{sigma_i}, the set $\Sigma$ is equal to \eqref{Sigma}.
\end{proof}

Using the input–state data~\eqref{data2}, we characterize a set of system matrices of the bilinear model that are consistent with the observed data. 
Since the bilinear system~\eqref{system_bilinear}, with a bounded residual term can represent the nonlinear system~\eqref{system_nonlinear}, we first aim to design a controller that stabilizes all admissible bilinear systems within this consistent set. 
This ensures that the true bilinear system~\eqref{system_bilinear} is also stabilized.

\begin{figure*}[htbp]
\begin{equation}\label{sos}
    \begin{bmatrix}
\begin{bmatrix} 
\begin{bmatrix}
H d(z)-g(z) I_N  &0\\
0 &0
\end{bmatrix}-M(\tau(z))&  &   \\
&  \frac{g(z)}{2c_x^2} I_N  &  \\
&   & \frac{g(z)}{2c_u^2} I_m  
\end{bmatrix}& 
\begin{bmatrix} 0 \\ Hd(z) \\ L(z) \\ L(z) \otimes z \\ Hd(z) \\ L(z) \end{bmatrix} &0\\
\begin{bmatrix} 0 & H d(z) & L^T(z) \!\!& (L(z) \otimes z)^T  & H d(z) \!\!& L^T(z) \end{bmatrix} 
& 
Hd(z) & Hd(z)\\
0 &Hd(z) &\rho d(z) I_N
\end{bmatrix} \in \text{SOS}[z]^{5N+2m+mN}_{2\alpha}
\end{equation}
\rule{\textwidth}{0.4pt} 
\vspace{-10pt}
\end{figure*}

\subsection{Data-driven Controller Design for Bilinear Systems}\label{sec3.2}

We start with the controller design for the bilinear system with all possible consistent system matrices.
According to \cite{vatani2014control}, an unstable discrete-time bilinear system can be stabilized by a rational state-feedback control law and neither a linear nor a polynomial state-feedback is able to stabilize the system.
Therefore, we propose a rational controller design method for the bilinear system with all possible system matrices $(A, B, \tilde{B})\in\Sigma$ and all possible residuals $r\in Re$ based on an SOS program in the following theorem.

\begin{theorem}\upshape\label{theorem1}
Suppose that Assumption~1 holds.
Given the data \eqref{data2}, if there exist $\alpha\in\mathbb{N}$, $H\in\mathbb{R}^{N\times N}$, $L\in\mathbb{R}[z]^{m\times N}_{2\alpha-1}$, $d\in \text{SOS}_+[z]_{2\alpha}$, $\tau_i\in\text{SOS}[z]^{T_f}_{2\alpha}, \forall 
i\in\mathbb{I}_{[1, T_f]}$, $g\in \text{SOS}_+[z]_{2\alpha}$, and $\rho>0$ such that the SOS condition \eqref{sos} holds with $z\in\mathbb{R}^N$, then the derived rational controller
\[u(\Phi(x))=\frac{1}{d(\Phi(x))}K(\Phi(x))\Phi(x)\]
with $K(\Phi(x))=L(\Phi(x))H^{-1}$ exponentially stabilizes the origin of the uncertain system $\Phi(x_{t+1})\!=\!A\Phi(x_t)\!+\!Bu_t\!+\!\tilde{B}(u_t\otimes \Phi(x_t))\!+\!r(x_t, u_t)$ for any $(A, B, \tilde{B})\in\Sigma$ and $r\in Re$.
\end{theorem}
\begin{proof}
For brevity, we write $(\Phi(x_{t+1}), \Phi(x_t), u_t, r(x_t, u_t))$ as $(z_+, z, u, r)$.
Substituting the rational control law $u = \frac{K(z)}{d(z)}z$ and using $\frac{K(z)}{d(z)}z\otimes z=(\frac{K(z)}{d(z)}\otimes z)z$, the unknown bilinear system dynamics $\Phi(x_{t+1})=A\Phi(x_t)+Bu_t+\tilde{B}(u_t\otimes \Phi(x_t))+r(x_t, u_t)$ with any consistent system matrices $(A, B, \tilde{B})\in\Sigma$ and any $r\in Re$ can be rewritten as
\begin{equation}\nonumber
z_+ = \left[ A + B\frac{K(z)}{d(z)} + \tilde{B} \left( \frac{K(z)}{d(z)} \otimes z \right) \right] z + r,
\end{equation}
which yields the structured form
\begin{equation}\nonumber
z_+ = \begin{bmatrix}A & B & \tilde{B}\end{bmatrix} \begin{bmatrix} I_N \\ \frac{K(z)}{d(z)} \\ \frac{K(z)}{d(z)} \otimes z \end{bmatrix} z + r
\end{equation}
Defining $z_{AB}=\begin{bmatrix} I_N \\ \frac{K(z)}{d(z)} \\ \frac{K(z)}{d(z)} \otimes z \end{bmatrix} z$ and $w_{AB} = \begin{bmatrix}A & B & \tilde{B}\end{bmatrix}z_{AB}$, we have
\begin{equation}\nonumber
z_+ = w_{AB} + r.
\end{equation}

Applying Schur complement to \eqref{sos} twice implies~\eqref{eq:proof1}. 
\begin{figure*}[!t]
\begin{equation}\label{eq:proof1}
\begin{bmatrix} 0 \\ Hd(z) \\ L(z) \\ L(z)\otimes z \\ Hd(z)  \\ L(z) \end{bmatrix}
    [Hd(z)-\frac{1}{\rho}H^2 d]^{-1}
    \begin{bmatrix} \star \end{bmatrix}^T 
    +
    \begin{bmatrix} 
\begin{bmatrix}
H d(z)\!-\!g(z) I_N  \!\!\!&0\\
0 \!\!\!&0
\end{bmatrix}\!\!-\!\!M(\tau(z))&  &   \\
&  \frac{g(z)}{2c_x^2} I_N  &  \\
&   & \frac{g(z)}{2c_u^2} I_m  
\end{bmatrix}\succeq 0
\end{equation}
\rule{\textwidth}{0.4pt}
\end{figure*}
Define $P=H^{-1}$ and $K(z)=L(z)H^{-1}$. Multiplying \eqref{eq:proof1} by $\frac{1}{d(z)}$,  
the resulting inequality is equivalent to~\eqref{eq:proof2}, 
\begin{figure*}[!t]
\begin{equation}\label{eq:proof2}
\begin{bmatrix} 
0 \\ 
I_N \\ 
\frac{K(z)}{d(z)} \\ 
\frac{K(z)}{d(z)}\otimes z \\ 
I_N \\ 
\frac{K(z)}{d(z)} 
\end{bmatrix}
(\frac{1}{\rho}I_N-P)^{-1}
\begin{bmatrix} \star \end{bmatrix}^T 
+ \begin{bmatrix} 
\begin{bmatrix}
P^{-1}-\frac{g(z)}{d(z)} I_N & 0\\
0 & 0
\end{bmatrix}
-\frac{1}{d(z)}M(\tau(z)) 
& & \\
& \frac{g(z)}{2c_x^2d(z)} I_N & \\
& & \frac{g(z)}{2c_u^2d(z)} I_m  
\end{bmatrix}
\succeq 0
\end{equation}
\hrulefill
\end{figure*}
which is further equivalent to
\begin{equation}\label{proof3}
   E\!\!
    \begin{bmatrix}
        (\frac{1}{\rho}I_N-P)^{-1} \hspace{-1.5em}& \hspace{-1.5em}& \hspace{-3.5em}&\\
        \hspace{-1.5em}& P^{-1} \hspace{-1.5em}& \hspace{-3.5em}& \\
        \hspace{-1.5em}& \hspace{-1.5em}&(S(\tau(z))d(z))^{-1} \hspace{-3.5em}&\\
        \hspace{-1.5em}& \hspace{-1.5em}& \hspace{-3.5em}&\frac{g(z)}{d(z)} \Pi_r^{-1}  
    \end{bmatrix}\!\!
    E^T \succeq 0,
\end{equation}
where $S(\tau(z)) = \begin{bmatrix} I & 0 \\ 0 & -I \end{bmatrix} M(\tau(z))^{-1} \begin{bmatrix} -I & 0 \\ 0 & I \end{bmatrix}$, $\Pi_r=\mathrm{diag}(-I_N, 2c_x^2I_N, 2c_u^2 I_m)$ and
\begin{equation}\nonumber
    E\!=\!\!\begin{bmatrix}
      0  \!\!\!\!&-I_N \!\!\!\!& I_N \!\!\!\!& 0 \!\!& I_N \!\!\!\!& 0 \!\!\!\!& 0 \\
      \begin{bmatrix}I_N\\ \frac{K(z)}{d(z)}\\ \frac{K(z)}{d(z)}\!\otimes\! z\end{bmatrix}  \!\!\!\!&0 \!\!\!\!& 0 \!\!\!\!& -I_{N+m+mN} \!\!\!\!& 0 \!\!\!\!& 0 \!\!\!\!& 0 \\
      I_N  \!\!\!\!&0 \!\!\!\!& 0 \!\!\!\!& 0 \!\!\!\!& 0 \!\!\!\!& -I_N \!\!\!\!& 0 \\
      \frac{K(z)}{d(z)}  \!\!\!\!&0 \!\!\!\!& 0 \!\!\!\!& 0 \!\!\!\!& 0 \!\!\!\!& 0 \!\!\!\!& -I_m 
    \end{bmatrix}.
\end{equation}
Using the Dualization lemma in \cite[Lemma 4.9]{scherer2000linear}, \eqref{proof3} is equivalent to 
\begin{equation}\label{lyapu}
    \tilde{E}
    \begin{bmatrix}
        \frac{1}{\rho}I_N-P\hspace{-1em}& \hspace{-1.5em}& \hspace{-1.5em}& 
        \\
        \hspace{-1.5em}&P \hspace{-1em}& \hspace{-1.5em}&\\
        \hspace{-1.5em}& \hspace{-1em}& S(\tau(z))d(z) \hspace{-1.5em}& \\
        \hspace{-1.5em}& \hspace{-1em}& \hspace{-1.5em}&\frac{d(z)}{g(z)}\Pi_r
    \end{bmatrix}
    \tilde{E}^\top \preceq 0,
\end{equation}
where
\[\tilde{E}=\begin{bmatrix} 
        I_N \!\!\!\!& 0 \!\!\!\!& 0 \!\!\!\!& \begin{bmatrix} I_N \\ \frac{K(z)}{d(z)} \\ \frac{K(z)}{d(z)} \otimes z \end{bmatrix}^T \!\!\!\!& 0 \!\!\!\!& I_N \!\!\!\!& \frac{K(z)^\top}{d(z)}  \\
        0 \!\!\!\!& I_N \!\!\!\!& 0 \!\!\!\!& 0 \!\!\!\!& I_N \!\!\!\!& 0 \!\!\!\!& 0 \\
        0 \!\!\!\!& I_N \!\!\!\!& I_N \!\!\!\!& 0\!\!\!\!& 0 \!\!\!\!& 0 \!\!\!\!& 0  
    \end{bmatrix}.\]
Left and right multiplying \eqref{lyapu} with $\begin{bmatrix}z^\top &r^\top &w_{AB}^\top\end{bmatrix}$ and its transpose, we obtain
\begin{equation}\label{lyapu2}\nonumber
    \begin{bmatrix}z\\ z_+ \\ w_{AB} \\ z_{AB} \\ r \\ z \\ u 
    \end{bmatrix}^\top\!\!\!\!\!\!
    \begin{bmatrix}
        \frac{1}{\rho}I_N-P \hspace{-1em}& \hspace{-1.5em}& \hspace{-1.5em}& 
        \\
        \hspace{-1.5em}&P \hspace{-1em}& \hspace{-1.5em}&\\
        \hspace{-1.5em}& \hspace{-1em}& S(\tau(z))d(z) \hspace{-1.5em}& \\
        \hspace{-1.5em}& \hspace{-1em}& \hspace{-1.5em}&\frac{d(z)}{g(z)}\Pi_r
    \end{bmatrix}\!\!
    \begin{bmatrix}z\\ z_+ \\ w_{AB} \\ z_{AB} \\ r \\ z \\ u 
    \end{bmatrix} \!\!\preceq \!0,
\end{equation}
which is further equivalent to
\begin{equation}\label{decay}
\begin{aligned}
    &- \|z\|_P^2\!+\frac{1}{\rho}\|z\|^2 \!+\! \|z_+\|_P^2 +\! \begin{bmatrix} w_{AB} \\ z_{AB} \end{bmatrix}^\top \!\!\!\!\!S(\tau(z)) d(z)\!\! \begin{bmatrix}w_{AB} \\ z_{AB} \end{bmatrix} \\
    &- \frac{d(z)}{g(z)}\|r\|^2 + \frac{2c_x^2 d(z)}{g(z)}\|z\|^2 + \frac{2c_u^2 d(z)}{g(z)}\|u\|^2
    \leq 0
\end{aligned}
\end{equation}
using $u=\frac{K(z)}{d(z)}z$.
Using the bound on the residual function in \eqref{residual}, we have
\begin{equation}\label{bound}\nonumber
\|r\|^2\leq (c_x\|z\|+c_u\|u\|)^2\leq 2c_x^2 \|z\|^2+2c_u^2 \|u\|^2,
\end{equation}
which further gives the following inequality using the multiplier $\frac{d(z)}{g(z)}> 0$
\begin{equation}\label{rzu}
    - \frac{d(z)}{g(z)}\|r\|^2 + \frac{2c_x^2 d(z)}{g(z)}\|z\|^2 + \frac{2c_u^2 d(z)}{g(z)}\|u\|^2\geq 0.
\end{equation}
Since all consistent system matrices $(A, B, \tilde{B})\in\Sigma$ satisfy $[I \ A \ B \ \tilde{B}] M(\tau(z)) [I \ A \ B \ \tilde{B}]^T \preceq 0$, by applying the Dualization Lemma, we obtain:
\begin{equation}\label{S(z)}
\begin{bmatrix} [A \ B \ \tilde{B}] \\ I \end{bmatrix}^T S(\tau(z)) \begin{bmatrix} [A \ B \ \tilde{B}] \\ I \end{bmatrix} \preceq 0.
\end{equation}
Multiplying \eqref{S(z)} from left and right by $z_{AB}^\top\sqrt{d(z)}$ and its transpose, respectively, using $w_{AB}=\begin{bmatrix}A &B &\tilde{B}\end{bmatrix}z_{AB}$, we have
\begin{equation}\label{wz}
    \begin{bmatrix}
        w_{AB}\\
        z_{AB}
    \end{bmatrix}^\top
    S(\tau(z))d(z)
    \begin{bmatrix}
        w_{AB}\\
        z_{AB}
    \end{bmatrix}\geq 0.
\end{equation}
Using \eqref{rzu} and \eqref{wz} in \eqref{decay}, we obtain
\begin{equation}\label{Lyapu-QR}
    \|z_+\|_P^2 - \|z\|_P^2 \le -\frac{1}{\rho} \|P\|_2^2\|z\|^2
\end{equation}
which holds for $z_+=Az+Bu+\tilde{B}(u\otimes z)+r$ with any $(A, B, \tilde{B})\in\Sigma$ and $r\in Re$.
Thus, the Lyapunov inequality \eqref{Lyapu-QR} implies that the uncertain system $\Phi(x_{t+1})\!=\!A\Phi(x_t)\!+\!Bu_t\!+\!\tilde{B}(u_t\otimes \Phi(x_t))\!+\!r(x_t, u_t)$ with any $(A, B, \tilde{B})\in\Sigma$ and $r\in Re$ is exponentially stabilized to the origin by the derived rational controller.
\end{proof}

\begin{remark}\upshape\label{remark_obj}
Theorem~\ref{theorem1} shows that, by solving the SOS program~\eqref{sos} using the collected data, a rational state-feedback controller can be derived that exponentially stabilizes the bilinear system for all consistent system matrices and all residual errors satisfying Assumption~\ref{assumption1}. 
This result is proved via the Lyapunov decay inequality~\eqref{Lyapu-QR}, which holds with decay rate $\frac{1}{\rho}$.
The decay rate can be further optimized by incorporating an objective function that minimizes $\rho$, which means maximizing $\frac{1}{\rho}$. 
A larger value of $\frac{1}{\rho}$ leads to faster exponential convergence of the system to the origin. 
Alternatively, $\rho$ can be fixed to achieve a desired convergence rate.
\end{remark}

\begin{remark}\upshape
The SOS program \eqref{sos} is linear in the decision variables $H, L, \tau, g$ and $\rho$ when the denominator $d$ is fixed.
Thus, it can be efficiently solved using convex optimization techniques for a given choice of $d$.
The degree $\alpha$ of the rational controller also plays an important role.
In particular, choosing a higher degree $\alpha$ provides greater flexibility in the controller design but increases the computational complexity.
If the SOS program \eqref{sos} is not feasible with the chosen denominator $d$, a possible strategy is to consider an alternative denominator $d$, increase the degree $\alpha$, or use a larger dataset.
\end{remark}

\subsection{Koopman-based Data-driven Controller Design for Nonlinear Systems}\label{sec3.3}

The previous results establish a controller design method for uncertain bilinear systems. 
We now extend this result to the continuous-time nonlinear systems. 
This extension is based on the result that the bilinear system~\eqref{system_bilinear} with residual $r$ satisfying Assumption~\ref{assumption1} can represent the nonlinear system~\eqref{system_nonlinear} over the regions $\mathbb{X}$ and $\mathbb{U}$.

In the following theorem, we develop a data-driven controller design method for the continuous-time nonlinear system~\eqref{system_nonlinear}.

\begin{theorem}\label{theorem2}\upshape
Suppose that Assumption~1 holds.
Given the data \eqref{data2}, if there exist $\alpha\in\mathbb{N}$, $H\in\mathbb{R}^{N\times N}$, $L\in\mathbb{R}[z]^{m\times N}_{2\alpha-1}$, $d\in \text{SOS}_+[z]_{2\alpha}$, $\tau_i\in\text{SOS}[z]^{T_f}_{2\alpha}, \forall 
i\in\mathbb{I}_{[1, T_f]}$, $g\in \text{SOS}_+[z]_{2\alpha}$, and $\rho>0$ such that the SOS program \eqref{sos} holds with $z\in\mathbb{R}^N$, then the sampled controller
\begin{equation}\nonumber
\begin{aligned}
u_s(x(k))=u(x(t\Delta k)), k\in[t\Delta k, (t+1)\Delta k], t\in\mathbb{N} 
\end{aligned}
\end{equation}
with $u(x)=\frac{1}{d(\Phi(x))}L(\Phi(x))H^{-1}\Phi(x)$ exponentially stabilizes the origin of the continuous-time nonlinear system~\eqref{system_nonlinear}.
\end{theorem}
\begin{proof}
As shown in Theorem~\ref{theorem1}, \eqref{sos} implies that \eqref{Lyapu-QR} holds for any $(A, B, \tilde{B})\in\Sigma$ and $r\in Re$.
Since $(A_s, B_s, \tilde{B}_s)\in\Sigma$, the Lyapunov decay also holds for the bilinear system~\eqref{system_bilinear}.
Then, we refer to \cite[Corollary 4.2]{strasser2026safedmd} to show that the closed-loop continuous-time nonlinear system~\eqref{system_nonlinear} is exponentially stabilized based on the sampled rational state-feedback control law.
\end{proof}

\begin{remark}\upshape
Theorem~\ref{theorem2} shows that the continuous-time nonlinear system~\eqref{system_nonlinear} is exponentially stabilized by the sampled controller designed for the bilinear model~\eqref{system_bilinear}. 
The controller is obtained by solving the SOS program~\eqref{sos} with a fixed denominator $d$ using the data~\eqref{data2}.
In general, Assumption~\ref{assumption1} can only be satisfied over suitable sets $\mathbb{X}$ and $\mathbb{U}$ with a certain probability. Consequently, closed-loop stability of the nonlinear system can only be guaranteed probabilistically for initial states that lie inside a subset of $\mathbb{X}$.
\end{remark}

\begin{remark}\upshape 
We previously proposed a data-driven min–max MPC scheme for bilinear systems
\cite{xie2025bilinear}. 
In contrast, this paper develops a controller design for nonlinear systems based on a bilinear model derived via Koopman operator theory.
The resulting bilinear model includes a residual term that is proportionally bounded by the lifting function and the input, whereas in \cite{xie2025bilinear} the bilinear model is corrupted by noise bounded by a constant. 
These differences lead to different data-driven characterizations of the system matrices and stability proofs.
Furthermore, the derived bilinear model in this work is formulated in terms of a lifted state, while \cite{xie2025bilinear} considers only the standard state. 
Compared with the Koopman-based controller design for nonlinear systems in \cite{strasser2025koopman}, which first identifies a bilinear model via EDMD and then designs a controller for the identified system, the present approach directly designs the controller using data collected from the nonlinear system.
\end{remark}

\section{Numerical Examples}\label{sec:4}
In this section, we demonstrate the effectiveness of the proposed data-driven controller design method for a nonlinear system based on its Koopman bilinearization.
The simulations are conducted in MATLAB using the SOSTOOLS \cite{sostools} and the MOSEK solver \cite{apsmosek} for SOS programs.

We consider a nonlinear system with the following continuous-time dynamics
\begin{equation}\label{sys_simul}
\dot{x}=x+(1+\sin(x))u.
\end{equation}
The sampling time is chosen as $\Delta k=0.01s$.
To facilitate the Koopman bilinearization, we choose the lifting function as $\Phi(x)=\begin{bmatrix}x &\sin(x)\end{bmatrix}^\top$.
The nonlinear system \eqref{sys_simul} can be represented by \eqref{system_bilinear} with unknown system matrices $A_s, B_s, \tilde{B}_s$.
Over the compact state region $\mathbb{X}=[-\pi/2, \pi/2]$ and input region $\mathbb{U}=[-2.5, 2.5]$, the residual error bounds in Assumption~\ref{assumption1} are determined as $c_x = 1.7\times 10^{-4}$ and $c_u=0.7\times 10^{-2}$, which are known.
To construct the consistent set $\Sigma$ defined in \eqref{Sigma}, we collect $T_f = 200$ input-state pairs. 
The initial states are sampled uniformly from $\mathbb{X}$, and the inputs are drawn randomly from $\mathbb{U}$.

We implement the proposed data-driven control scheme on the nonlinear system~\eqref{sys_simul}.
The initial state is chosen as $x_0=1$.
For the rational state-feedback controller, we choose the degree $\alpha=1$.
For the controller design, we fix the denominator by including all monomials in $\mathbb{R}[\Phi(x)]_2$, i.e., $d=1+x^2+x\sin(x)+\sin(x)^2$, which is strictly SOS.
Fig.~\ref{fig:1D} illustrates the closed-loop state and input trajectories of the resulting nonlinear system.
The results show that the designed rational state-feedback controller stabilizes the unknown nonlinear system, with both state and input trajectories converging exponentially to the origin.

\begin{figure}
  \centering
  \includegraphics[width=0.95\linewidth]{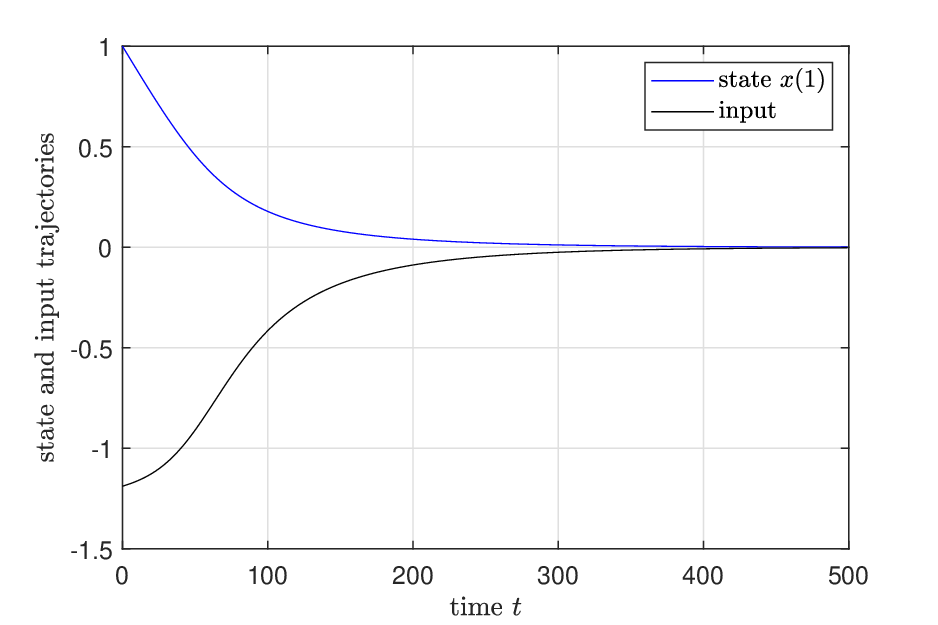}
  \caption{Closed-loop state and input trajectories of the nonlinear system~\eqref{sys_simul}.}\label{fig:1D}
\end{figure}

To further investigate the performance of the controller, we add an objective function to maximize the Lyapunov decay rate, i.e., $\max \frac{1}{\rho}$ (or equivalently $\min \rho$), as discussed in Remark~\ref{remark_obj}.
To compare the influence of the degree of the controller, we choose two different degrees, i.e., $\alpha=1$ or $\alpha=2$.
For each degree, the denominator $d$ is fixed by including all monomials in $\mathbb{R}[\Phi(x)]_{2\alpha}$.
All other parameters remain the same as the previous simulation.
Fig.~\ref{fig:1Dcompare} illustrates the closed-loop state and input trajectories when choosing different $\alpha$ and adding the objective function.
The result shows that the trajectories in Fig.~\ref{fig:1Dcompare} converge to the origin much faster than that in Fig.~\ref{fig:1D}.
Maximizing the decay rate significantly accelerates the convergence speed of the trajectories.
Furthermore, increasing the degree to $\alpha=2$ provides additional degrees of freedom in the SOS program, resulting in faster convergence of the trajectories. 
However, it also leads to a higher computational complexity: when $\alpha=2$ the computational time is $39.85$ seconds, while when $\alpha=1$, the computational time is $15.24$ seconds.
This illustrates a trade-off between closed-loop performance and computational tractability over $\alpha$.

\begin{figure}
  \centering
  \includegraphics[width=0.95\linewidth]{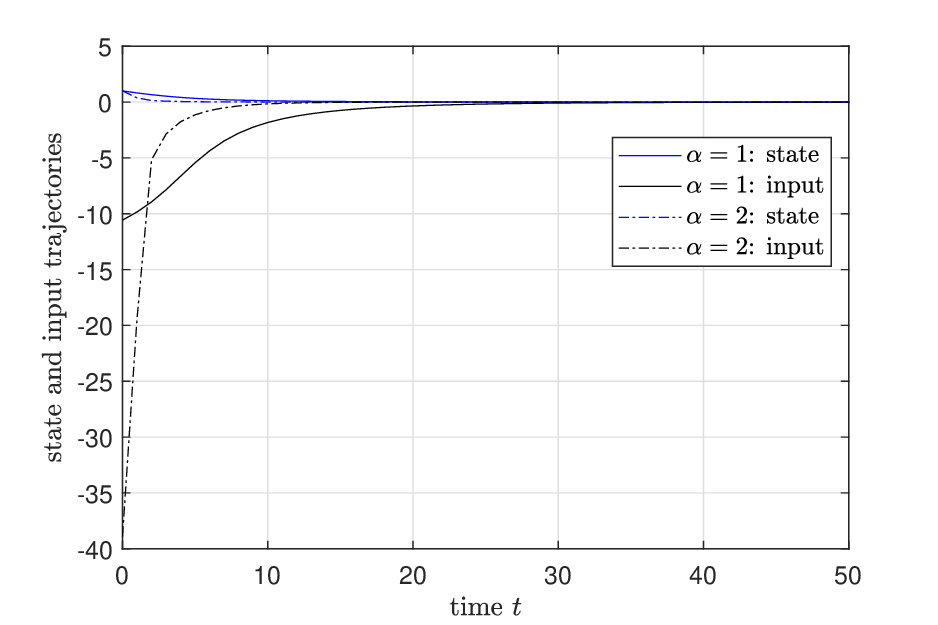}
  \caption{Closed-loop state and input trajectories of the nonlinear system~\eqref{sys_simul} choosing different degree $\alpha$ and adding objective function.}\label{fig:1Dcompare}
\end{figure}



\section{Conclusion}\label{sec:5}
In this paper, we propose a direct data-driven controller design method for continuous-time nonlinear systems based on a Koopman bilinear realization. 
Using sampled data collected from the nonlinear system, we derive a data-driven characterization of the set of consistent bilinear systems.
We then formulate an SOS program to derive a rational state-feedback control law. 
We prove that the nonlinear system is exponentially stabilized at the origin by a sampled controller designed for the bilinear system. 
Furthermore, numerical examples demonstrate the effectiveness of the proposed approach on a nonlinear system.





\bibliographystyle{IEEEtran} 
\bibliography{references} 

\end{document}